\newif\ifLIPICS
\LIPICSfalse
\ifx\Start\undefined\else\Start\fi
\ifLIPICS
\documentclass[a4paper,UKenglish]{lipics}
\serieslogo{}
\volumeinfo
  {Billy Editor, Bill Editors}
  {2}
  {Conference title on which this volume is based on}
  {1}
  {1}
  {1}
\EventShortName{}
\DOI{10.4230/LIPIcs.xxx.yyy.p}

\else
\documentclass[12pt,a4paper]{elsarticle}
\usepackage[a4paper,margin=2cm,bmargin=3cm]{geometry}
\usepackage{amsmath}
\usepackage{amssymb}
\usepackage{url}
\usepackage[utf8]{inputenc}

\newtheorem{definition}{Definition}
\newtheorem{example}{Example}
\newtheorem{theorem}{Theorem}
\newtheorem{lemma}{Lemma}
\newtheorem{proposition}{Proposition}
\newtheorem{proof}{Proof}
\fi

\usepackage{comment}
\DeclareFontFamily{OT1}{rsfs}{}
\DeclareFontShape{OT1}{rsfs}{m}{n}{ <4-6> rsfs5 <6-9> rsfs7 <9-> rsfs10 }{}
\DeclareMathAlphabet{\mathscript}{OT1}{rsfs}{m}{n}

\DeclareMathAlphabet{\mathbi}{OML}{cmm}{bx}{it}

\ifx\jcbemptyset\undefined\let\emptyset=\varnothing
\else\let\emptyset=\jcbemptyset\fi

\def\rest{\mathord{\upharpoonright}}
\newcommand{\rs}{$R$-stopped}
\newcommand{\bs}{$B$-stopped}

\newcommand{\es}{\mathcal{E}}
\newcommand{\confs}{\mathcal{V}}
\newcommand{\un}{\mathcal{U}}
\newcommand{\gpn}{\mathcal{G}}

\newcommand{\unf}{unfolding\-$^-$}
\newcommand{\occ}{occurrence\-$^-$}

\newcommand{\prob}{\mathbb{P}}

\makeatletter
\newcommand{\dotminus}{\mathbin{\text{\@dotminus}}}

\newcommand{\@dotminus}{%
  \ooalign{\hidewidth\raise1ex\hbox{.}\hidewidth\cr$\m@th-$\cr}%
}
\makeatother

\def\checkstop{\afterassignment\checkstopaux\let\next= }
\def\checkstopaux{\ifx s\next\def\next{}\fi\next}
\def\abbres/{\textsc{es}\checkstop}
\def\abbrees/{\textsc{ges}\checkstop}
\def\abbrses/{\textsc{ses}\checkstop}
\def\abbrpes/{\textsc{pes}\checkstop}
\def\abbrpres/{\textsc{p}{\small r}\textsc{es}\checkstop}

\def\Thetahat{\hat\Theta}
\ifLIPICS
\newtheorem{conjecture}[theorem]{Conjecture}
\newtheorem{proposition}[theorem]{Proposition}
\fi
\usepackage{tikz}

\begin{document}


\ifLIPICS
\else
\begin{frontmatter}
\fi

\title{On probabilistic stable event structures}

\author{Nargess Ghahremani and Julian Bradfield}

\ifLIPICS
\else
\let\affil=\address
\fi  

\affil{Laboratory for Foundations of Computer Science, University of
Edinburgh,\\ 10 Crichton St, Edinburgh, EH8 9AB, U.K.\\ \texttt{jcb@inf.ed.ac.uk}}

\ifLIPICS
\Copyright{N. Ghahremani and J. C. Bradfield}


\maketitle              

\fi

\ifLIPICS
\else

\end{frontmatter}
\fi

\thispagestyle{plain}
\section{Introduction}

Verification has been a core concern of theoretical computer science
for several decades. As most real systems are in some sense
distributed or concurrent, one has to decide how to model
concurrency. Broadly, there are two schools: the interleaving approach
imposes a global `clock' on the system, and says that independent
concurrent events occur in some order, but that order is arbitrary --
encapsulated by the CCS interleaving law $a | b = a.b + b.a$. `True
concurrency', on the other hand, represents concurrency explicitly in
the semantics of the system, as exemplified by Petri nets. There are
many arguments about the merits, ranging from the practical
-- e.g.\ that true concurrent models avoid the
exponential state explosion arising from arbitrary interleaving -- to
the philosophical -- one should understand concurrency in its own
right.

Another extension of classical computation is probability, whether
used to model genuine randomness, or uncertainty. Probabilistic models
have in the last two or three decades also been the subject of
intensive research, and by now there are many widely used tools for
practical verification of probabilistic systems, based on a range of
different theoretical models. A small selection of
established models and systems might be
\cite{rabin1963probabilistic,puterman2009markov,vardi1985automatic,pnueli1993probabilistic,yi1992testing,segala1995probabilistic,de1997formal,baier1998algorithmic}. Such work adopts the
`interleaving' approach to concurrency -- that is, from the true
concurrency proponent's point of view, it ignores concurrency.

The topic of this paper is the combination of probability with true
concurrency. As well as the
core interleaving vs.\ true concurrency distinction, 
there is another crucial difference between the two in the probabilistic
framework: in probability, one cares about the probabilistic
(in)dependence of events, and in true concurrency this relation is (at
least) correlated with, and ideally derived from, the concurrency
relation. To put it another way,
temporal stochastic processes and models capturing concurrency
through nondeterminism such as
\cite{kartson1994modelling,de1998stochastic,rabin1963probabilistic,segala1996modeling,derman1970finite}
have a global state corresponding to a global time. In a distributed
system, however, this is neither feasible nor natural. Thus, in true
concurrency approaches there is no notion of global time or state, but
rather local ones. In other words, the local components have their own
local states and act in their own local time until they communicate
together. This results in a highly desirable match between concurrency and
probability, so that concurrent
choices can be made probabilistically independent.

There is some work, quite recently, on probabilistic true
concurrency models
\cite{volzer2001randomized,haar2002probabilistic,varacca2003probabilistic,varacca2004probabilistic,benveniste2003markov,AB06,abbes2008true},
but this is still at the foundational semantic stage, and there are
very few no true concurrent probabilistic temporal logics -- a recent
example in the setting of distributed Markov chains is in \cite{Jha15}.

Therefore, in \cite{thesis} the first author aimed to develop a
probabilistic temporal logic suitable for application to Petri nets,
which are one of the best known true concurrent models, even though
often used with a forced interleaving semantics. As part of this
programme, it was necessary first to extend lower-level models, which
is the work reported here.

\subsection{Related work}

Event structures \cite{W86}, and extensions thereof, are our
basic model of concurrency. While they are a very `low
level' model, unsuited for direct modelling of systems, they model
directly concurrency and causality -- causality being (roughly) the converse of concurrency.

To our knowledge, the first probabilistic model for event structures
 was given in \cite{katoen1996quantitative}, which
 defines probabilistic extended bundle event
structures. As with the stable event structures we discuss later, 
these allow different possible causes for an event, only one of
which can be the cause in a run. However, they capture choice by means
of
groups of events which are mutually in conflict and enabled at the
same time, called \emph{clusters}. Clusters can be determined
statically, and maintain the concept of choice internal to the system.

A domain theoretic view, closely related to the probabilistic
powerdomains of \cite{jones1989probabilistic,tix2009semantic}, was
taken by Varacca, Volzer and Winskel \cite{varacca2004probabilistic},
who define continuous valuations on the domain of configurations. They
then define \emph{non-leaking valuations with independence} on
confusion-free event structures. These turn out to coincide with
distributed probabilities as defined by \cite{AB06}.

Randomised Petri nets along with their corresponding probabilistic
branching processes defined in \cite{volzer2001randomized}  focus on
free-choice conflicts only and therefore do not deal with confusion. 

Finally, the most recent probabilistic event structures defined by Winskel \cite{winskel2013distributed} extend existing notions of probabilistic event structures in order to make them suitable for dealing with certain interactions between strategies.

The main hurdle in defining a probabilistic concurrent
system with partial order semantics is to find units of
choice so that concurrent units are probabilistically
independent. Abbes and Benveniste in \cite{AB06} define distributed
probabilities taking \emph{branching cells} as units of choice and
show that  in probabilistic event structures not only does concurrency
match probabilistic independence, but also that this cannot be
achieved at a grain finer than that of branching cells. Furthermore,
they show how finite configurations can be decomposed into branching
cells dynamically, where maximal configurations of the branching
cells enforce all the conflicts within the cell to be
resolved. \emph{Local  probabilities} are assigned to each branching
cell and these can be extended to a \emph{limiting
probability} measure on the space of maximal configurations. The only
constraint required is that  of \emph{local finiteness} which can be
viewed as bounded confusion and which is defined later.

\subsection{Summary}

In this article, we first explain the motivation for studying
probabilistic stable event structures. Then, after the necessary
preliminaries, we define \emph{conflict-driven} stable event 
structures, a superclass of event structures allowing enough
`confusion' to model Petri nets (the ultimate aim of this work),
but constrained enough to allow a manageable probabilistic semantics.
To extent the concepts of \cite{AB06} to this setting,
we need to
constrain them to a certain sub-class of \emph{jump-free} (stable)
event structures. We then show that probabilistic jump-free stable event structures can be defined analogously to probabilistic event structures. 

\section{Motivation}

As we said above, our original aim was probabilistic logic for Petri
nets, and this work forms part of the route to there. So before diving
in to the unfortunately but inevitably highly technical development of
probabilistic stable event structures, let us explain the key issues in
the  more widely understood formalism of Petri nets.

Consider the following two safe nets:

\begin{figure}[h]
\centering
\def\m{label=center:$\bullet$}
\begin{tikzpicture}
[place/.style={inner sep=1.6mm,circle,thick,draw=black},
mplace/.style={inner sep=1.6mm,circle,thick,draw=black,label=center:$\bullet$},
transition/.style={inner sep=2mm,rectangle,thick,draw=black}]

\node (p1) at (1,1) [mplace] {};
\node (p2) at (3,1) [mplace] {};
\node (p3) at (5,1) [mplace] {};

\node (t1) at (0,2) [transition,label=left:$t_1$] {};
\node (t2) at (2,2) [transition,label=left:$t_2$] {};
\node (t3) at (4,2) [transition,label=left:$t_3$] {};
\node (t4) at (6,2) [transition,label=left:$t_4$] {};

\node (p5) at (1,3) [place] {};
\node (p6) at (4,3) [place] {};

\node (t5) at (1,4) [transition,label=left:$t_a$] {};
\node (t6) at (4,4) [transition,label=left:$t_b$] {};

\draw [->] (p1) to (t1);
\draw [->] (p1) to (t2);
\draw [->] (p2) to (t2);
\draw [->] (p2) to (t3);
\draw [->] (p3) to (t3);
\draw [->] (p3) to (t4);

\draw [->] (t1) to (p5);
\draw [->] (t2) to (p5);
\draw [->] (p5) to (t5);
\draw [->] (t3) to (p6);
\draw [->] (p6) to (t6);

\node (s1) at (8,3) [transition,label=$u_1$] {};
\node (s2) at (10,3) [transition,label=$u_2$] {};
\node (s3) at (9,3) [transition,label=$v$] {};

\node (q1) at (8,2) [mplace] {};
\node (q2) at (10,2) [mplace] {};

\draw [->,bend right=20] (q1) to (s1);
\draw [<-,bend left=20] (q1) to (s1);
\draw [->,bend right=20] (q2) to (s2);
\draw [<-,bend left=20] (q2) to (s2);

\draw [->] (q1) to (s3);
\draw [->] (q2) to (s3);
\end{tikzpicture}
\end{figure}
The left net demonstrates several instances of confusion, where the
occurrence of one event changes the possibilities elsewhere in the
net. Now suppose that we add probabilities. There are several ways to
do this, but they will all boil down to assigning probabilities to the
choices between say $t_1$ or $t_2$, and $t_2$ or $t_3$, and so on.
However, because of confusion, the various choices are not
independent. For example, whether a choice arises between $t_a$ and
$t_b$ depends on whether $t_b$ was enabled by $t_1$ or $t_2$. Worse,
an event arbitrarily far away may be the source of such confusion: the
firing of $t_4$ disables $t_3$, thereby stopping it from conflicting
with $t_2$. Worse still, this confusion then propagates to affect the
choice between $t_a$ and $t_b$, which superficially have nothing to do
with each other. In
the absence of independence, it is hard to see how to give a useful
probabilistic semantics that reflects concurrency. Thus, it is equally
hard to extend probabilistic temporal logics (making statements such
as `transition $t$ fires with probability ${}>0.5$ in all futures') to
nets.

The right net is perhaps even more pathological: $u_1$ and $u_2$ can
fire away independently, but the malign presence of $v$ waiting to
disable them both at some time means that they cannot be treated as
truly independent.

This problem, which more generally we might characterize as `the
problem of negative causality', has been troubling the community for
decades, and we have not solved it. Instead we looked for ways to
analyse differing levels of confusion, so that at least we could
define manageable classes of nets. In the course of this, the first
author \cite{thesis} first developed a new notion of `compact
unfolding' of Petri net, which allows some degree of backward conflict
to be maintained without being unfolded out. It turns out that just as
standard net unfoldings give rise to an event structure semantics
for Petri nets, compact unfoldings give rise to a stable event
structure semantics. Hence we needed to extend the existing work on
probabilistic semantics from event structures to stable event structures.

\section{Preliminaries}

Owing to the lengthy formal definitions needed for the work we build
on, we shall present most of the preliminaries in summary form,
referring to \cite{thesis} or standard texts for formalities.

\subsection{Event structures and their variants}

The event structures of Winskel \cite{W86} model concurrency
by describing \emph{consistency} and \emph{enabling} relations between
\emph{events}. They are a low-level model, as each event corresponds
to an occurrence of a transition or action in higher-level models such
as Petri nets or process calculi. Thus,
they are suited to foundational investigations in concurrency theory.

For historical reasons, the terminology is a little
confusing, and has also changed during the development of the
subject. Here we shall use \emph{general event structures} for the
original and most general formulation; \emph{stable event structures}
impose the 
restriction that an event has a unique set of causing events; and
plain \emph{event structures} are those generated by
binary \emph{causality} and \emph{conflict} relations on events, rather
than more general predicates. Owing to our very frequent use of these
long terms, we will abbreviate them.

\begin{definition}\label{ees}\label{ees:conf}
A \emph{general event structure (\abbrees/)} $\es$ is a set of
\emph{events} $E$ equipped with a non-empty subset-closed
\emph{consistency predicate} $Con \subseteq \wp_{\mathrm{fin}}(E)$ on
finite sets of events, and an \emph{enabling relation} ${\vdash}
\subseteq Con \times E$, monotone in the first argument.
A \emph{configuration} of $\mathcal{E}$ is a (possibly infinite)
subset $x\subseteq E$ that is finitely consistent, and such that every
$e\in x$ has a finite enabling set in $x$.
The set of all configurations of $\mathcal{E}$ is represented by $\confs(\es)$ or $\confs_\es$ or $\confs$ when no confusion arises.
\end{definition}

\begin{definition}
Let $(P,\sqsubseteq)$ be a partial order. Then $S\subseteq P$ is
\emph{compatible} (${S\uparrow}$) iff $\exists p \in P.\;\forall s\in
S.\;s\sqsubseteq p$. A subset is finitely compatible (${S\uparrow}^{\textit{fin}}$) iff $\forall S_0 \subseteq_{\textit{fin}}S.\;{S_0\uparrow}$.
\end{definition}

The configurations $\mathcal{V}=\confs(\es)$ of $\es$ form a family of
subsets of $E$ that is finite-complete ($\mathcal{A} \subseteq \mathcal{V}
\;\&\; {\mathcal{A}\uparrow} ^{\textit{fin}} \Rightarrow \bigcup \mathcal{A}\in \mathcal{V}$), finitely based ($\forall u\in \mathcal{V}.\; \forall e\in
u.\; \exists v\in \mathcal{V}.\;(v \text{ is finite } \&\; e \in v \;\&\;
v\subseteq  u)$), and coincidence-free ($\forall u\in \mathcal{V}.\;
\forall e,e' \in u.\; e\neq e' \Rightarrow (\exists v\in \mathcal{V}.\;
v\subseteq u \;\&\;(e\in v \Leftrightarrow e' \notin v))$).
It can be shown that given a family $\mathcal{V} \subseteq 2^E$ that is
finite-complete, finitely based and coincidence-free, $\mathcal{V}$
is the set of configurations of the \abbrees/
$\mathcal{E}=(E,Con,\vdash)$ given by $x \in Con
\Leftrightarrow_{\textit{def}} x \text{ is finite }\&\; \exists u \in
\mathcal{V}.\; x\subseteq u$ and $x \vdash e \Leftrightarrow_{\textit{def}} x \in Con \;\&\; \exists u \in \mathcal{V}.\; e\in u \;\&\; u \subseteq x\cup\{e\}$.

General event structures are one of the most general classes of
event structures; they allow for an event to have different causes,
which is a desirable property. However, problems arise when dealing
with configurations in which an event does not have a unique cause, in
that its different causes can occur at the same time. Such
configurations can be excluded by applying a \emph{stability}
constraint, leading to the definition of stable event structures. It
is worth noting that for stable event structures there is no global
partial order of causal dependency on events, but each configuration
has its own local partial order of causal dependency. 

\begin{definition}\label{ses}
A \emph{stable event structure (\abbrses/)}, is an \abbrees/ such that
mutually consistent enabling sets are closed under intersection, so
that any set enabling $e$ has a unique minimal enabling subset, the
`causes' of $e$.

Non-empty compatible configurations of an \abbrses/ are closed under
intersection, and such a \emph{stable family} of configurations induces an \abbrses/.

Given a stable family $\mathcal{V}$ of configurations, and
$e,e'\in u \in \mathcal{V}$ define 
$e \leq_u e'$ if for all $v:\,\mathcal{V} \ni v\subseteq u$ we have $e'\in
v$ implies $e\in v$ (i.e.\ $e$ is a necessary member of the
$u$-history of $e'$).
Let $\lceil e\rceil_u$ be the least configuration $v \subseteq u$
containing $e$ (the `causal history' of $e$ in $u$).
\end{definition}

An alternative approach to stability is to define $\leq$
in the structure as a global partial order of causal dependency on
events: 

\begin{definition}\label{pes}
  A \emph{prime event structure (\abbrpes/)} $\mathcal{E}$ on $E$
  comprises $Con$ as before, and a partial order $\leq$ on $E$ (the
  \emph{causality relation}) such that the down-closure $\lceil
  e\rceil$ of $e$ under $\leq$ is finite, singleton events are
  consistent, $Con$ is closed under subsets and under adding causes of
  events.
  The set of configurations of a \abbrpes/ is the set of
  $\leq$-down-closed consistent configurations, and is a stable family
  of configurations for $\mathcal{E}$ \cite{W86}.
\end{definition}

Thus a \abbrpes/ is itself a \abbrses/ with the
naturally induced enabling relation. Less obviously, 
a \abbrses/ can be translated into a \abbrpes/
by extending and renaming the events, so that each event is augmented
with the history of how it occurred in a configuration. Then, even
though the events and as such the configurations are different, the
domains of configurations of both event structures are isomorphic as
partial orders \cite{W86}. This translation is indeed right adjoint to
the embedding of \abbrpes/ in \abbrses/, in the categorical
framework that Winskel sets up.

Finally, generating the consistency relation from a
binary conflict relation gives

\begin{definition}\label{es}
An \emph{event structure (\abbres/)} $\mathcal{E}$ on $E$ 
comprises a causality relation $\leq$ as above, and a binary,
symmetric and irreflexive \emph{conflict relation} $\#$ on $E$ which is
is $\leq$-upwards-closed (i.e.\ events inherit the conflicts
of their causes).


A configuration of $\mathcal{E}$ is a $\leq$-downwards-closed and
conflict-free subset of $E$. Write $\lceil e\rceil$ as above, and
$\lceil x\rceil$ for $\bigcup_{e\in x}\lceil e\rceil$.

\end{definition}

\subsection{Probabilistic Event Structures}
\label{sec:prob-es}

This subsection is a (severe) summarization of the concepts introduced
by \cite{AB06}. The abstract definition is conceptually fairly simple.

A \emph{probability measure} on a space $\Omega$ is a
function from a suitable collection (technically a $\sigma$-algebra)
of subsets of $\Omega$ to $[0,1]$, satisfying the
appropriate behaviour for probabilities (the measure of a union of
disjoint sets is the sum of the measures).

Now consider an \abbres/, and let $\Omega$ be the space of its maximal
configurations (finite or infinite, perhaps uncountably infinite)
where the system has run to completion. There is a $\sigma$-algebra on
$\Omega$ comprising the sets $\{ \omega \in \Omega : \omega \supseteq
v\}$ for each configuration $v$ -- that is, each set is a full subtree
of the configuration tree rooted at $v$. Call this set $S(v)$ (the
\emph{shadow} of $v$). The intuition is that the probability of $v$ is
equal to the sum of the probabilities of all the maximal
configurations $\omega$ reachable from $v$ -- except that not all such
configurations are probabilistically independent, so the usual
probability calculations for non-independent events have to be made.

\begin{definition}\label{defn:prob-es}
A \emph{probabilistic event structure (\abbrpres/)} is a pair $(\es, \prob)$, where
$\prob$ is a \emph{probability measure} on the space of maximal
configurations of $\es$ with the $\sigma$-algebra of shadows.

The \emph{likelihood function} on configurations of $\es$ is the
function $p(v) = \prob(S(v))$.
\end{definition}

This definition is very abstract, and hard to calculate with or
implement via a process calculus. Much of the work of \cite{AB06} is
in giving an alternative operational presentation, which we now
introduce via a long series of notions. In the rest of this section,
$\es = (E,{\leq},\#)$ is an \abbres/, $e, e'$ etc.\ range over $E$,
and $u, v$ etc.\ range over $\mathcal{V}(\es)$.

\begin{definition}
 $P \subseteq
 E$ is a \emph{prefix} if $P= \lceil P \rceil$ (closed under causes).
If $F \subseteq E$, then $(F, {\leq} {\rest} F, \# {\rest} (F \times F))$ is an \abbres/ (which is a sub-\abbres/ of $\es$).
We use $F$ to also refer to the sub-\abbres/ induced by $F$, when no confusion arises. 

Hence, configurations are conflict-free prefixes.
$v$ is \emph{maximal} iff it is a $\subseteq$-maximal
  configuration.
We denote the maximal configurations of $\es$ by $\Omega(\es)$ or just
$\Omega$.

The \emph{future} $\es^{v}= (E^{v},\leq^{v},\#^{v})$ of 
$v$, where $E^v = \{e \notin v : \lceil e\rceil \cup v \in \mathcal{V}\}$, is the \abbres/ generated by events that can happen
after $v$. Given $u \in \mathcal{V}$ and $v \in \mathcal{V}^u$,
the \emph{concatenation} $u \oplus v$ is just $u \cup v$; and given
configurations $v\subseteq u$, the  \emph{subtraction} $u
\ominus v$ is just $u \setminus v$.
\end{definition}

In a probabilistic event structure, the probability reflects the
notion of choice which arises whenever  a conflict is encountered for
the first time. Thus, we consider the `first-hand' conflicts,
i.e.~conflicts which are not inherited, as constituents of units of
choice. Then we consider those prefixes which contain all immediate
conflicts, and maximal configurations within them, which resolve all
choices that can be made. The final step is to close under
concatenation to give a notion of `R-stopped configuration', which
provides a probabilistically independent unit.

\begin{definition}
Events $e$ and $e'$ are in \emph{immediate conflict} $e \#_\mu e'$ iff
$ \# \cap (\lceil e \rceil \!\times\! \lceil e'\rceil) = \{(e,e')\}$
(i.e.\ their conflict is not from conflicting causes).

A prefix $B$ of $\es$ is a \emph{stopping prefix} if it is
$\#_\mu\text{-closed}$.
Given a set $X \subseteq E$, $X^*$ is the closure of $X$ under
$\#_\mu$ and $\lceil\hbox{--}\rceil$ -- it is the \emph{minimal
  stopping prefix} containing $X$. 

$v$ is  \emph{\bs} if $v$ is  maximal
in $\mathcal{V}(B)$; and $v$ is
\emph{stopped} if it is $v^*$-stopped.

$v$ is \textit{R-stopped} if there is a sequence 
$\emptyset \subseteq v_1 \subseteq\dots \subseteq  v_{n} \dots$
such that $v = \bigcup v_n$ and each $v_{n+1}\ominus v_n$ is finite and stopped in $\es^{v_n}$.

The sequence is called a \textit{valid decomposition} of $v$ and if
$v = v_N$ for some $N$, then $v$ is \textit{finite R-stopped}.
The set of $R$-stopped configurations is written ${\mathcal{W}(\mathcal{E})}$. 
\end{definition}

Thus R-stopped configurations are the end of a sequence of steps,
each of which resolves all possible choices before it. It remains to
decompose these steps into their concurrent (i.e.\ independent) components.

\begin{definition}\label{Pre-regular Event Structures}
$\es$ is \textit{pre-regular}, if for every finite
$v$, there are finitely many events
enabled by $v$.

$\mathcal{E}$ is \textit{locally finite} if every $e \in E$ is in some
finite stopping prefix.
\end{definition}

  \cite{AB06} theorem 3.12 shows that then all
maximal configurations are \rs.
Henceforth all \abbres/s are locally finite.

\begin{definition}\label{initSP}
If $\mathcal{X} \subset \mathcal{V}$, write $\overline{\mathcal{X}}$
for the subset of its finite configurations.

An initial stopping prefix is a minimal non-empty stopping prefix.

\label{def:Branching-Cell}\label{def:bc-enabled}
A \emph{branching cell}
\emph{enabled by} a finite \rs\  $v$ is an initial stopping prefix
of $\es^{v}$. The set of such branching cells is notated $\delta_{\es}(v)$.

It can be shown that $v$, with valid decomposition $(v_n)$, has a
\emph{covering} $\Delta_{\es}(v)$, a 
sequence $(c_n)$ of branching cells such that $v_n$ enables $c_{n+1}$
and $v_{n+1}\setminus v_n$ is maximal in $c_{n+1}$. That is, $(c_n)$
gives a sequence of maximal independent steps to reach $v$.

The set of all branching cells of $\es$ is denoted by
$\mathcal{C}(\es)$ and the set of maximal configurations of a
branching cell $c$ is denoted by $\Omega_c$.  

\end{definition}

It is important to note that while branching cells decomposing a configuration are disjoint, in general, different branching cells of an \abbres/ may overlap. This is because not all of the events that can potentially constitute a branching cell are enabled at every configuration. Thus, configurations determine their corresponding branching cells and in this way we say that decomposition through branching cells is dynamic. 

Now the operational definition of \abbrpres/'s is 
given by equipping each branching cell of an \abbres/ with a
probability for its maximal configurations; \cite{AB06} show
(highly non-trivially) that the operational and abstract definitions give the
same class of structures.

\begin{definition}
A locally finite \abbres/ $\mathcal{E}$ is \emph{locally
randomised} if each branching cell $c \in \mathcal{C}$ is equipped
with a \emph{local transition probability} $q_c$ on the set $\Omega_C$
of configurations that can be chosen within the cell.
\label{Probabilistic Event Structure}
The likelihood function $p:\overline{\mathcal{W}}\rightarrow [0,1]$ is defined as:
$$\forall v\in \overline{\mathcal{W}}, \text{    }p(v)= \prod_{c\in {\Delta}(v)} q_{c}(v \cap c)$$
where $\Delta(v)$ denotes the covering of $v$ in $\mathcal{E}$.

$p$ induces a probability measure
$\prob_B$ on the (countable) space $B$ of stopping prefixes of $\es$;
and the full probability measure $\prob$ can be derived from $\prob_B$
via a construction called the \emph{distributed product} and the
Prokhorov extension theorem. Hence the measure $\prob$ is called the
distributed product of the branching probabilities $\{ q_c : c \in C \}$.

\end{definition}

Further, this definition of likelihood
function also applies to the space of $R$-stopped
configurations of $\es$, and so \cite{AB06} obtains a definition of probabilistic event
structures in which local choice probabilities are attached either to
branching cells or to $R$-stopped configurations.

\subsection{Categorical notions}
It is convenient to have to hand a few items from Winskel's
categorical toolkit:

\begin{definition}\cite{W86}
Let $\es_0=(E_0, Con_0, \vdash_0)$ and $\es_1=(E_1, Con_1, \vdash_1)$ be two \abbrses/s. A \emph{morphism} from $\es_0$ to $\es_1$ is a partial function $\theta: E_0 \rightarrow E_1$ on events satisfying:
\begin{enumerate}
\item $X \in Con_0 \Rightarrow \theta.X \in Con_1$
\item $\{e,e'\} \in Con_0 \;\&\; \theta(e)=\theta(e') \Rightarrow e=e'$
\item $X \vdash_0 e \;\&\; \theta(e) \text{ is defined }\Rightarrow \theta.X \vdash_1\theta(e)$
\end{enumerate}
A morphism is \emph{synchronous} if it is a total function.
\end{definition}
As remarked earlier, there is an inclusion $I$ from \abbrpes/s to \abbrses/s, with an adjoint functor
back. We will need this functor:

\begin{definition}\label{def:trans-ses-pes-}
Given a \abbrses/ $\es_0=(E, Con,\vdash)$, let $\Theta(\es_0)$ be the \abbrpes/ $\es_1=(P, Con_P,\leq)$, with isomorphic domain of configurations, defined as follows. 
\begin{enumerate}
\item $P=\{\lceil e \rceil_x ~|~ e\in x \in \confs(\es) \}.$
\item $p' \leq p \Leftrightarrow p' \subseteq p.$
\item $X \in Con_P \Leftrightarrow X\subseteq_{\text{fin}} P ~\&~ X\uparrow .$
\end{enumerate}
$\Theta$ maps morphisms thus: if $\theta\colon \es \to \es'$ (induced by
an event map $\theta:E \to E'$), then $\Theta(\theta)$ is $\{\lceil e \rceil_x ~|~ e\in x \in
\confs(\es) \} \mapsto \{\lceil \theta(e) \rceil_{x'} ~|~ \theta(e)\in
x' \in \confs(\es') \}$.

For a \abbrses/ $\es_0$ we refer to $\Theta(\es_0)$ as its associated \abbrpes/. 

The counit of the adjunction is $\theta\colon \es_0 \mapsto \theta_{\es_0}$,
where $\theta_{\es_0}\colon \Theta(\es_0)\to\es_0$ is the synchronous morphism (of \abbrses/) given by
$\theta_{\es_0}(p)=e \text { for }p= \lceil e \rceil_x \in P,\; e \in
E \;\&\; x \in \confs(\es_0)$.
\end{definition}

\section{Conflict-driven (Stable or Prime) Event Structures}\label{sec:conflict-driven}

The difficulty in extending probabilistic notions to rich concurrent
structures such as \abbrses/s lies in the consistency
relation, and ensuring that causal (in)dependence matches
appropriately with probabilistic (in)dependence. As we have just seen,
it is already quite technically intricate for the relatively simple
case of event structures with a binary conflict relation. Our
contribution here is to develop the framework further to allow more
sophisticated consistency relations. In particular, we will define a
class of \abbrses/s which is sufficient to give a
low-level semantics for Petri nets, in which consistency arises not
just out of immediate conflicts, but out of the history of previous
conflicts; and moreover there is the possibility of \emph{confusion},
where the execution of \emph{prima facie} concurrent events is interfered with
by other events.

In the definition of \abbrses/s, the consistency predicate $Con$ is required to satisfy only one condition. Namely, $Y \subseteq X \;\&\;X \in Con \Rightarrow Y \in Con$. This does not necessarily have to fit with the configurations of the \abbrses/s. Consider the following example. 

\begin{example}
Let $\es=(E,Con,\vdash)$ be a \abbrses/, where
$E=\{e_1,e_2,e_3,e_4\}$, $\{e_1\} \vdash e_2, \{e_3\}\vdash e_4$,
$\{e_1,e_3\} \notin Con$ and $\{e_2,e_4\} \in Con$. It is easy to see
that even though $e_2$ and $e_4$ are consistent, they can never appear
together in a configuration. Therefore, the consistency predicate is
not sensible with respect to the configurations. (In terms introduced
just below, the consistency relation here is only recording the
immediate conflicts, not the inherited conflicts.)
\end{example}

We therefore define \emph{sensible} \abbrses/s  and \abbrpes/s as follows. 

\begin{definition}\label{def:conflict-driven-ses}
Let $\es$ be a \abbrses/ or \abbrpes/ with the consistency relation
$Con$. We say $\es$ is \emph{sensible} iff $\forall X \in Con
\Leftrightarrow \exists v \in \confs(\es).\; X \subseteq v$. If $\es$
is not sensible, it can be made so by pruning $Con$ of the unreachable
consistent sets.
\end{definition}

So far we have described how the consistency predicate in sensible
structures relates to the configurations of that structure. We
now define the notions of conflict and immediate conflict.

\begin{definition}
Two events $e$ and $e'$ of a \abbrses/ are in \emph{conflict} under a finite configuration $v$, represented by $e \#_v e'$ iff $\big(\{e\} \cup \{e'\} \cup v \big)\notin Con$. Then two events are in conflict, represented by $e \# e'$, iff $\forall v \in \mathcal{V}.\;e \#_v e'$. 


\label{def:imm-conf--}
Define \emph{immediate conflict} 
between two events $e$ and $e'$ of a \abbrses/ under configuration $v$: 
$$e \#_{\mu ,v} e' \Leftrightarrow_{def} v \vdash e\;\&\;  v \vdash e' \;\&\; \#_{v} \cap (\lceil e \rceil_v \times \lceil e' \rceil_v)  = \{(e,e')\}.$$
and define $e \#_\mu e'$ iff $\forall v,v' \in \mathcal{V}.\; v \vdash
e \;\&\; v' \vdash e' \Rightarrow \exists v'' \subseteq v \cup v'.\; e
\#_{\mu ,v''} e'$. 
\end{definition}

\begin{definition}
For a set $X \subseteq_{\text{fin}} E$, let ${}_*X$ be the set of the
sets consisting of exactly one history $\lceil e \rceil_v$ for each
event $e$ in $X$ and configuration $v \ni e$. (That is, for every
configuration, ${}_*X$ contains a single choice among all the
histories that can have produced each event in $X$.)
\end{definition}

We now define a special class of \abbrses/s, namely,
conflict-driven \abbrses/s. Originally, their definition
arose through considering unfoldings of Petri nets, and so in
\cite{thesis} they are called `net-driven', but here we abstract away
from the net derivation.

\begin{definition}
A \abbrses/ $\es=(E, Con, \vdash)$ is called \emph{conflict-driven} iff it satisfies the following.
\begin{enumerate} 
\item $\es$ is sensible.
\item 
$\forall X \subseteq_{fin} E.\; X \notin Con \Rightarrow \forall T\in{}_*X.\;  \exists e_1, e_2 \in   \bigcup T.\; e_1 \#_{\mu} e_2$
\item $\forall e, e' \in E, v \in \confs.\; e \#_{\mu,v} e' \Rightarrow e \# e'$
\end{enumerate}
Note that from 2 and 3 it follows that $\forall X \subseteq E.\; X \notin Con \Rightarrow \forall T \in {}_*X.\; \exists e_1, e_2 \in  \bigcup T.\; e_1 \# e_2$
\end{definition}
As mentioned before, the first characteristic describes that the
consistency predicate is in line with configurations and the second
one implies that the source of inconsistency is a conflict in the
past. The last constraint describes the persistence of conflicts 
(originally because immediate conflict in a Petri net is a cause of
later conflict).

We can now show that for the associated \abbrpes/ of a
conflict-driven \abbrses/, the consistency predicate
can be generated from a binary conflict relation:

\begin{theorem}\label{thm:ses2pes}
Let $\es_0=(E,Con,\vdash)$ be a conflict-driven \abbrses/ and let $\es_1=\Theta(\es_0)=(P,Con_P,\le)$ be its associated \abbrpes/. Then,
 we have:
$$X \in Con_P \Leftrightarrow X \subseteq_{\text{fin}} P \;\&\; \forall p,p' \in X. \neg (p \# p')$$
where $p\#p'$ iff $\{p,p'\}\notin Con_P$.
\end{theorem}

If a \abbrpes/ does have a consistency relation generable
from a binary conflict relation, it can be seen as an \abbres/ via an inclusion mapping $\hat I$.
%
%
Thus a conflict-driven \abbrses/ generates an
\abbres/. 
\begin{definition}\label{def:hatted}
Given a conflict-driven \abbrses/ $\es$, we denote by $\hat\es$ or
$\Thetahat(\es)$  the \abbres/ ${\hat I}(\Theta(\es))$, and we refer
to $\hat\es$ as the associated \abbres/ of $\es$. Similarly we write
$\tilde\theta_\es$ for the adjunct morphism $\hat I(\Theta(\es)) \to
\Theta(\es) \to \es$.
\end{definition}

\section{Probabilistic Jump-free Stable Event Structures}
We now consider adjoining probabilities to \abbrses/s. First we
present the definition of concepts analogous to those of probabilistic
event structures.  Then our aim is to derive an isomorphism between the events of branching cells of conflict-driven \abbrses/s and 
their associated \abbres/; we find that such isomorphisms exist if the \abbrses/s are \emph{jump-free}, as we shall define.

{\fontseries{b}\selectfont We assume that the {\rm\abbrses/s} in this section are conflict-driven unless stated otherwise.}

\subsection{Branching Cells on Stable Event Structures}
We now define branching cells for \abbrses/s, in a
similar manner and show that in general, unlike the branching cells of
\abbres/s, they do not form the units of choice.

\begin{definition}\label{Prefix}\label{Configuration}
A subset $P \subseteq E$ is called a \textit{prefix} of a \abbrses/ $\es$ iff $ \forall e\in P.\;\exists X \subseteq P.\; X \vdash e$.

Let $\es = (E,Con,\vdash)$ be a \abbrses/ and let $F$ be a prefix of
$E$. Then $(F, \{ X \cap F : X \in Con\}, \{ (X\cap F,e\in F) :
X\vdash e \})$ is a \abbrses/ (which is a sub-\abbrses/ of $\es$). We use
$F$ also for the sub-\abbrses/ induced by $F$, when no confusion arises.
Configurations can then be viewed as consistent prefixes as before.
Concepts of compatibility of configuration and maximal
configurations are defined as for \abbres/s and 
we represent the set of maximal configurations of \abbrses/
$\mathcal{E}$ by $\Omega(\es)$.

In this setting, 
the \emph{future} of a configuration $v$ of $\es$ is $\es^{v}=(E^{v},Con^v,\vdash^v)$ where
$E^{v}=\{ e \notin v : \{e\}\cup v \in Con \}$
and $Con^{v}, {\vdash^{v}}$ are the natural restrictions to
$E^v$. (Note that the future of $v$ includes all events that might
happen, both those directly enabled by $v$ and those completely independent.)
\end{definition}

The notions of $B$-stopped and stopped configurations for \abbrses/s are similar to those of \abbres/s. However,
unlike \abbres/s, given $X$ a subset of events, a canonical
stopping prefix including $X$ cannot be derived. This is because in
the definition of prefix for \abbrses/s an event can have
different sets of events enabling it. Thus, these  notions are defined
as follows, recalling the definition of $\#_{\mu, v}$ (definition \ref{def:imm-conf--}).

\begin{definition}\label{Stopping Prefix}\label{Stopped Configuration}
A prefix $B$ of a \abbrses/ $\es$ is called a \textit{stopping prefix}
if it is $\#_ {\mu,v}\text{-closed}$ in the following sense: 
$$\forall v \subseteq B.\; e \in v \;\&\; \exists e' \in E.\; e \#_{\mu,v}e' \Rightarrow e' \in B$$

A configuration $v$ of $\es$ is called \textit{$B$-stopped} if $v$ is a maximal configuration of $B$; $v$ is called \textit{stopped} if there is a stopping prefix $B$ such that $v$ is $B$-stopped.
\end{definition}

Stopping prefixes of a conflict-driven \abbrses/s have
a close relation with the stopping prefixes of their associated 
\abbres/. To expand this further, we first observe the relation
between $\#_{\mu,v}$ of a \abbrses/ and $\#_\mu$ of its
corresponding \abbres/ by the following immediate lemma:

\begin{lemma}\label{lem:imm-conf-2}
Given a  \abbrses/ $\es$ and $\hat\es = \Thetahat(\es)$, then $e \#_{\mu,v} e' \Leftrightarrow \lceil e \rceil_v \#_\mu \lceil e' \rceil_v$
\end{lemma}

Using this lemma we can describe the relation between \abbrses/ stopping
prefixes and \abbres/ stopping prefixes as follows. 

\begin{proposition}
Given a \abbrses/ $\es$, then  $B$ is a stopping prefix of $\es$ iff $\Thetahat(B)$ is a stopping prefix of $\hat\es$. 
\end{proposition}
\begin{proof}
It is easy to verify that  $B$ is a prefix iff $\hat B=\Thetahat(B)$
is a prefix as well (follows from the definition of $\Thetahat$ and
$\tilde\theta_\es$ (definitions \ref{def:trans-ses-pes-}, \ref{def:hatted}) being a morphism).
Also, from lemma \ref{lem:imm-conf-2} it follows that $B$ is
$\#_{\mu,v}$-closed (for $v \subseteq B$) iff $\hat B$ is $\#_\mu$-closed.
\end{proof}

The following are the appropriate adaptations of the \abbres/ definitions:

\begin{definition}\label{R-stopped-Configurations}
A configuration $v$ of \abbrses/ $\es$ is \textit{R-stopped} if there is a non-decreasing sequence of configurations $(v_n)$ for $0\leq n<N\le\infty$ such that:
\begin{enumerate}
\item $v_{0}=\emptyset$ and $v=\bigcup_{0{\leq}n<N}v_n$ , and
\item $\forall n{\geq}0, \text{    }n+1<N \Rightarrow v_{n+1} \ominus v_n$ is finite stopped in $\es^{v_n}$.
\end{enumerate}
The sequence is called a \textit{valid decomposition} of $v$ and if
$N<\infty$ then $v$ is said to be \textit{finite R-stopped}. The set
of $R$-stopped configurations of an \abbrses/ $\es$ is denoted ${\mathcal{W}(\es)}$.

\label{Pre-regular Event Structures-}
$\mathcal{E}$ is \textit{pre-regular}, if for every finite configuration $v$ of $\mathcal{E}$, the set $\{e \in E \mid v \oplus \{e\}\}$ is finite.

\label{Locally Finite Event Structures-}
$\mathcal{E}$ is \textit{locally finite} if for every $e \in E$, there is a finite stopping prefix of $\mathcal{E}$ containing $e$. 

As before, an \emph{initial stopping prefix} is a minimal non-empty stopping prefix.

\label{Branching Cell-}
A \emph{branching cell} of $\es$ and configuration $v \in
\overline{\mathcal{W}}({\es})$ is an {initial stopping prefix} of
$\es^{v}$. The set of all branching cells of $\es$ is denoted by
$\mathcal{C}(\es)$ and the set of maximal configurations of a branching cell $c$ is denoted by $\Omega_c$. 

The branching cells which are initial stopping prefixes of $\es^{v}$
are called the branching cells \emph{enabled} by $v$, and
denoted by $\delta_{\es}(v)$ or $\delta(v)$ if no confusion
arises. 
\end{definition}

\subsection{Probabilistic Event Structures and Stable Event Structures}
We are now ready to add probability to \abbrses/s. This
would be fairly straightforward, if the
branching cells of conflict-driven \abbrses/s and their
associated \abbres/s were isomorphic. However, 
the following example shows why this is not the case.

\begin{example}\label{ex:non-jump-free}
Consider the below \abbrses/ $\es$, corresponding to the Petri net at the
beginning of the paper, and its associated \abbres/ $\hat\es$, where
the dashed curved lines represent immediate conflicts.
The dotted
curved line shows immediate conflict under a particular configuration;
in this case $e_a \#_{\mu,\{e_1\}}e_b$.

\begin{figure}[h]
\hbox{\vbox{\hsize=0.5\textwidth
\centering 
\begin{tikzpicture}
[line width=1pt,event/.style={rectangle,fill=black,minimum width=3mm,
minimum height=0.5mm, inner sep=0pt}]
\node at (1,1) [event,label=right:$e_1$] (e1) {};
\node at (2,1) [event,label=right:$e_2$] (e2) {};
\node at (3,1) [event,label=right:$e_3$] (e3) {};
\node at (4,1) [event,label=right:$e_4$] (e4) {};
\node at (1.5,2) [event,label=right:$e_a$] (e6) {};
\node at (3,2) [event,label=right:$e_b$] (e7) {};
\draw [->] (e1) -- (e6); \draw [->] (e2) -- (e6); \draw [->] (e3) -- (e7);
\draw[dotted] (e6) to [bend left=40] (e7);
\draw[dashed] (e1) to [bend right=40] (e2);
\draw[dashed] (e2) to [bend right=40] (e3);
\draw[dashed] (e3) to [bend right=40] (e4);
\draw[dashed,line width=1.5pt,draw=lightgray] (0.3,0.3) -- (5,0.3) -- (5,2.7) --
(0.3,2.7) -- cycle;
\node at (1,2.7) [above] {$\tilde\theta_\es(c_1)$};
\end{tikzpicture}

{A (conflict-driven) stable event structure $\es$}
}\vbox{\hsize=0.5\textwidth
\centering 
\begin{tikzpicture}
[line width=1pt,event/.style={rectangle,fill=black,minimum width=3mm,
minimum height=0.5mm, inner sep=0pt}]

\draw[dashed,line width=1.5pt,draw=lightgray] (0.3,0.3) -- (0.3,2.7) --
(1.4,2.7) -- (1.4,1.4) -- (2.6,1.4) -- (2.6,2.7) -- (5,2.7) --
(5,2.7) -- (5,0.3) -- cycle;
\draw[dashed,line width=1.5pt,draw=lightgray] (1.6,2.7) -- (2.4,2.7) --
(2.4,1.6) -- (1.6,1.6) -- cycle;

\node at (1,1) [event,label=right:$e_1$] (e1) {};
\node at (2,1) [event,label=right:$e_2$] (e2) {};
\node at (3,1) [event,label=right:$e_3$] (e3) {};
\node at (4,1) [event,label=right:$e_4$] (e4) {};
\node at (1,2) [event,label=right:$e_a$] (e6) {};
\node at (2,2) [event,label=right:$e_a'$] (e6') {};
\node at (3,2) [event,label=right:$e_b$] (e7) {};
\draw [->] (e1) -- (e6); \draw [->] (e3) -- (e7);
\draw [->] (e2) -- (e6');
\draw[dashed] (e6) to [bend left=40] (e7);
\draw[dashed] (e1) to [bend right=40] (e2);
\draw[dashed] (e2) to [bend right=40] (e3);
\draw[dashed] (e3) to [bend right=40] (e4);

\node at (1,2.7) [above] {$c_1$};
\node at (2,2.7) [above] {$c_2$};
\end{tikzpicture}

{The corresponding \abbres/ $\hat\es=\Thetahat(\es)$.}
}}\end{figure}
\noindent As it can be seen from the figures, $\hat\es$ has two branching
cells ($c_1$ and $c_2$), while $\es$ (which ${} =
\theta_\es(\hat\es)$) has only one. To see 
this, let us construct the branching cells of
$\hat\es$. Consider the configuration $v_1=\{e_2,e_4,...\}$. Having event
$e_2$ implies that events $e_1,e_3,e_4,e_b$ and $e_a$ must be
added to the branching cell because of $\#_\mu$-closure. Thus,
$e_a'$ is not in this branching cell, but in the next branching cell,
consisting of $e_a'$ only. However, in the \abbrses/, both $e_a$ and
$e_a'$ of $\hat\es$ are represented by event $e_a$ of $\es$. Therefore, it
is not possible to cover $\es$ in any manner that is consistent with
the covering for $\hat\es$.  

Although in this example $c_2$ does not reflect any choice being made, more complex examples exist where all branching cells have a choice to make. Therefore, it is not possible to resolve this at the probabilistic level, e.g.~by trying to combine a number of branching cells with respect to their probabilities. 
\end{example}

In order to achieve isomorphic branching cells and avoid the above
situation (and more complex problems), we consider a class of
\abbrses/s (and their associated \abbres/s) which forbid these cases, namely, those structures which are \emph{jump-free}.


\begin{definition}\label{def:jump-free-es}
An \abbres/ is \emph{jump-free} iff $$\forall e,e'.\; e < e' \Rightarrow \nexists e_1,\ldots,e_k.  k>1 \mathrel{\&}e \#_\mu e_1, e_i \#_\mu e_{i+1}  \mathrel{\&} e_k \#_\mu e'$$ for  $1 \leq i \leq {k}$.
\label{def:jump-free-ses}
A \abbrses/ is \emph{jump-free} iff $$\forall e,e'. e <_v e' \Rightarrow \nexists e_1,\ldots,e_k.  k>1 \mathrel{\&} e \#_{\mu,v_0} e_1, e_i \#_{\mu,v_{i}} e_{i+1}, e_k \#_{\mu,v_{k}}e'$$ for $1 \leq i \leq {k-1} \mathrel{\&} v_i\subseteq v $.
\end{definition}

For example, 
the \abbrses/ in example \ref{ex:non-jump-free} is not jump-free as either of the chains of events $e_3,e_4,e_b$  and $e_1,\ldots,e_b$ break  jump-freeness. 

Jump-free \abbres/s are simpler to deal with, as, unlike \abbres/s, they are \emph{flat} in the following sense.

\begin{proposition}\label{prop:jfes}
The branching cells of jump-free \abbres/s (as initial stopping
prefixes) consist of initial events only (and the converse holds also).
\end{proposition}
\begin{proof}
Suppose $c$ has non-initial events and let $e \in c$ be such that  $\exists e_0 \in c.\; e_0<e \;\&\; \nexists e_1 \in c.\; e <e_1$. Then there exists an initial event $e'$ s.t. $\nexists e_0 \in c.\; e_0 < e' \;\&\; e' < e$. Noting that $c$ is an initial stopping  prefix and therefore, $\nexists c'. \;c' \subset c$, then in the formation of $\{e'\}^*$, $e$ can only be added to achieve the closure of $\#_\mu$. Therefore, there must be a chain of events $e_1, \ldots, e_k$ s.t. $e' \#_\mu e_1, e_i \#_\mu e_{i+1} \;\&\; e_k \#_\mu e$. Observe that  $k>1$ as otherwise $\neg (e \#_\mu e_1)$. Therefore, above chain forms a jump which is a contradiction and as such $c$ only consists of non-initial events. 

The converse follows immediately from branching cells being closed under $\#_\mu$.
\end{proof}

The same holds for \abbrses/s:
\begin{proposition}
The branching cells of jump-free \abbrses/s (as initial stopping prefixes) consist of initial events only.
\end{proposition}
\begin{proof}
The proof follows a similar reasoning to that of \abbres/s, noting that for all the initial events in $\es^v$, $\#_\mu$ is resolved meaning $\forall v' \in \confs(\es^v).\;e\#_{\mu, v'}e' \Rightarrow e \#_\mu e'$. 
\end{proof}

Recalling that the configurations of a \abbrses/ and its associated \abbres/ are isomorphic, we show in the following theorem that the branching cells of a jump-free conflict-driven \abbrses/ and its associated \abbres/ are isomorphic. 

\begin{theorem}\label{thm:bc-iso-bc-}
Given a jump-free conflict-driven \abbrses/ $\es$ and its associated
\abbres/ $\hat\es=\Thetahat(\es)$, $\hat{\mathcal{C}}$, the set of
branching cells of $\hat\es$, is isomorphic to $\mathcal{C}$, the
set of branching cells of $\es$. 
\end{theorem}

The most important consequence of
theorem \ref{thm:bc-iso-bc-} is that the covering of any configuration
in a \abbrses/ is exactly the same as that of its corresponding
configuration in its associated \abbres/. That is because the
configurations in the future of two isomorphic configurations are also
isomorphic and therefore, two isomorphic configurations have
isomorphic coverings. Therefore, all the probabilistic properties of
branching cells of \abbres/s are applicable to those of \abbrses/s,
and as such, all the probabilistic machinery described in section
\ref{sec:prob-es} for \abbres/s can be applied to conflict-driven
jump-free \abbrses/s.  
Thus, for example, the likelihood function for a \abbrses/ $\es$, $p:\overline{\mathcal{W}}\rightarrow \mathbb{R}$ is defined as:$$\forall v\in \overline{\mathcal{W}}, \text{    }p(v)= \prod_{c\in {\Delta}(v)} q_{c}(v \cap c)$$

\section{Conclusion}
We have introduced a new class of \abbrses/ called
conflict-driven \abbrses/s, which includes those \abbrses/ that arise
from Petri nets under the first author's `compact unfoldings'.

We then proceeded to extend the results of \cite{AB06} to \abbrses/,
finding that this is not possible in general, but that it is possible
for our new class of `jump-free' \abbres/ and \abbrses/ for which the
branching cells consist of events which 
are not causally related. We then proved that  for such stable event
structures and their associated  event structures, the branching cells
are isomorphic. Thus, probabilistic jump-free \abbrses/s
were defined in a similar manner to probabilistic event structures of
\cite{AB06}. 

The jump-free notion can be translated back to Petri nets, where it
means, more or less, that confusion is allowed provided that it is not
directly propagated forward by the causality relation. While the
jump-free structures are a larger class than the free-choice
structures, they are still far from complete: both the example nets in
our motivation section have jumps. We hope in future work to find
weakenings of the jump-free constraint.

\subsection*{Acknowledgements}

The first author was supported by a studentship from the Laboratory
for Foundations of Computer Science, and by a University of Edinburgh
Overseas Research Scholarship. Both authors were partly supported by
EPSRC grant EP/G012962/1. We are grateful to readers of \cite{thesis}
and earlier versions of this paper for very helpful comments.

\bibliographystyle{plain}
\bibliography{refs}

\begin{thebibliography}{10}

\bibitem{AB06}
Samy Abbes and Albert Benveniste.
\newblock True-concurrency probabilistic models branching cells and distributed
  probabilities for event structures.
\newblock {\em Inf. Comput.}, 204(2):231--274, 2006.

\bibitem{abbes2008true}
Samy Abbes and Albert Benveniste.
\newblock True-concurrency probabilistic models: Markov nets and a law of large
  numbers.
\newblock {\em Theoretical Computer Science}, 390(2):129--170, 2008.

\bibitem{baier1998algorithmic}
Christel Baier.
\newblock On algorithmic verification methods for probabilistic systems.
\newblock {\em Universit{\"a}t Mannheim}, 1998.

\bibitem{benveniste2003markov}
Albert Benveniste, Eric Fabre, and Stefan Haar.
\newblock Markov nets: probabilistic models for distributed and concurrent
  systems.
\newblock {\em Automatic Control, IEEE Transactions on}, 48(11):1936--1950,
  2003.

\bibitem{de1997formal}
Luca De~Alfaro.
\newblock {\em Formal verification of probabilistic systems}.
\newblock PhD thesis, Standford University, 1997.

\bibitem{de1998stochastic}
Luca De~Alfaro.
\newblock {\em Stochastic transition systems}.
\newblock Springer, 1998.

\bibitem{derman1970finite}
Cyrus Derman.
\newblock {\em Finite state Markovian decision processes}.
\newblock Academic Press, Inc., 1970.

\bibitem{thesis}
Nargess Ghahremani.
\newblock {\em Petri Nets, Probability and Event Structures}.
\newblock PhD thesis, University of Edinburgh, 2014.

\bibitem{haar2002probabilistic}
Stefan Haar.
\newblock Probabilistic cluster unfoldings for petri nets.
\newblock 2002.

\bibitem{Jha15}
Sumit~Kumar Jha, Madhavan Mukund, Ratul Saha, and P.~S. Thiagarajan.
\newblock Distributed markov chains.
\newblock In {\em Proc. 16th Int. Conf. on Verification, Model Checking, and
  Abstract Interpretation (VMCAI)}, 2015.

\bibitem{jones1989probabilistic}
Claire Jones and Gordon~D Plotkin.
\newblock A probabilistic powerdomain of evaluations.
\newblock In {\em Logic in Computer Science, 1989. LICS'89, Proceedings.,
  Fourth Annual Symposium on}, pages 186--195. IEEE, 1989.

\bibitem{kartson1994modelling}
D.~Kartson, G.~Balbo, S.~Donatelli, G.~Franceschinis, and G.~Conte.
\newblock {\em Modelling with generalized stochastic Petri nets}.
\newblock John Wiley \& Sons, Inc., 1994.

\bibitem{katoen1996quantitative}
Joost-Pieter Katoen.
\newblock {\em Quantitative and qualitative extensions of event structures}.
\newblock PhD thesis, University of Twente, 1996.

\bibitem{pnueli1993probabilistic}
Amir Pnueli and Lenore~D Zuck.
\newblock Probabilistic verification.
\newblock {\em Information and computation}, 103(1):1--29, 1993.

\bibitem{puterman2009markov}
Martin~L. Puterman.
\newblock {\em Markov decision processes: discrete stochastic dynamic
  programming}, volume 414.
\newblock Wiley. com, 2009.

\bibitem{rabin1963probabilistic}
Michael~O. Rabin.
\newblock Probabilistic automata.
\newblock {\em Information and control}, 6(3):230--245, 1963.

\bibitem{segala1996modeling}
Roberto Segala.
\newblock Modeling and verification of randomized distributed real-time
  systems.
\newblock 1996.

\bibitem{segala1995probabilistic}
Roberto Segala and Nancy Lynch.
\newblock Probabilistic simulations for probabilistic processes.
\newblock {\em Nordic Journal of Computing}, 2(2):250--273, 1995.

\bibitem{tix2009semantic}
Regina Tix, Klaus Keimel, and Gordon Plotkin.
\newblock Semantic domains for combining probability and non-determinism.
\newblock {\em Electronic Notes in Theoretical Computer Science}, 222:3--99,
  2009.

\bibitem{varacca2003probabilistic}
Daniele Varacca and Mogens Nielsen.
\newblock Probabilistic petri nets and mazurkiewicz equivalence.
\newblock 2003.

\bibitem{varacca2004probabilistic}
Daniele Varacca, Hagen V{\"o}lzer, and Glynn Winskel.
\newblock Probabilistic event structures and domains.
\newblock In {\em CONCUR 2004-Concurrency Theory}, pages 481--496. Springer,
  2004.

\bibitem{vardi1985automatic}
Moshe~Y. Vardi.
\newblock Automatic verification of probabilistic concurrent finite state
  programs.
\newblock In {\em Foundations of Computer Science, 1985., 26th Annual Symposium
  on}, pages 327--338. IEEE, 1985.

\bibitem{volzer2001randomized}
Hagen V{\"o}lzer.
\newblock Randomized non-sequential processes.
\newblock In {\em CONCUR 2001—Concurrency Theory}, pages 184--201. Springer,
  2001.

\bibitem{W86}
Glynn Winskel.
\newblock Event structures.
\newblock In {\em Advances in Petri Nets}, pages 325--392, 1986.

\bibitem{winskel2013distributed}
Glynn Winskel.
\newblock Distributed probabilistic strategies.
\newblock In {\em 29th Conference on the Mathematical Foundations of
  Programming Semantics}, 2013.

\bibitem{yi1992testing}
Wang Yi and Kim~Guldstrand Larsen.
\newblock Testing probabilistic and nondeterministic processes.
\newblock In {\em PSTV}, volume~12, pages 47--61, 1992.

\end{thebibliography}
\newpage
\section*{Appendix}

\begin{proof}[Proof of Theorem \ref{thm:ses2pes}]
($\Rightarrow$) follows from the definition of consistency relation. More precisely, if $\exists p,p'\in X.\;p \# p'$ then it follows that $p$ and $p'$ are not consistent (by definition of $\#$) and therefore,  they are not compatible as configurations of $\es_0$. Thus, $X \notin Con_P$ which is a contradiction. 

($\Leftarrow$) follows from the definition of conflict-driven \abbrses/s. More specifically, suppose by contradiction that there is a finite set of events $X$ s.t. $\forall p,p' \in X.\; \neg(p\#p')$ and $X \notin Con_P$. Let $X=\{p_i \mid p_i = \lceil e_i \rceil_{v_i} \;\&\; v_i \in \confs(\es_0)\}$. Since $X \notin Con_P$ this implies that $p_i$ as configurations of $\es_0$ are not compatible, in other words, if we let $\bar{X} = \cup \lceil e_i \rceil_{v_i}$, then $\bar{X} \notin Con$. Now since $\es_0$ is conflict-driven, then for any $T \in {}_*\bar{X}.\;\exists e_1,e_2 \in \bigcup T.\; e_1\#e_2$. For such $e_1, e_2$, suppose $p_1, p_2 \in X.\; p_1 = \lceil e\rceil_{v_1} \;\&\; e_1 \in p_1 \;\&\; p_2 = \lceil e' \rceil_{v_2} \;\&\; e_2 \in p_2$. Then $p_1$ and $p_2$ are not compatible as configurations of $\es_0$, therefore $\{p_1,p_2\} \notin Con_P$, which contradicts $\forall p,p'\in X.\; \neg(p \# p')$. Therefore, $X$ must be consistent, i.e.\;$X\in Con_P$.
\end{proof}

The proof of theorem \ref{thm:bc-iso-bc-} requires a few trivial
lemmas.

\begin{lemma}The following facts are immediate from the relevant
  definitions:

\abbrpes/s are sensible.

\label{lem:sens}
Let $\es_0$ be a sensible \abbrses/ and let $(P,Con_P,\le) = \Theta(\es)$. Then
$$ e \leq_x e' \Leftrightarrow \lceil e \rceil_x \subseteq \lceil e' \rceil_x$$
$$ X \in Con \Leftrightarrow \forall v \in \confs(\es_0) \;s.t.\; X \subseteq v.\; \{\lceil e \rceil_v \mid e \in X \} \in Con_P$$
\end{lemma}

\begin{lemma}\label{lem:sens-}
Let $\es$ be a conflict-driven \abbrses/ with associated \abbres/ $\hat\es$. Then we have $$e \#_{\mu,v} e' \Leftrightarrow \lceil e \rceil_v \#_\mu \lceil e' \rceil_v.$$
\end{lemma}
\begin{proof}
Follows trivially from lemma \ref{lem:sens} and definitions of $\#_{\mu,v}$ and $\#_\mu$.
\end{proof}

\begin{lemma}
Let $\es$ be a jump-free conflict-driven \abbrses/. Then $\hat\es$ is jump-free.
\end{lemma}
\begin{proof}
Follows trivially from lemmas \ref{lem:sens} and \ref{lem:sens-}.
\end{proof}

We now establish the desired isomorphism of branching cells. First:

\begin{lemma}\label{lem:morph-conf}
Let $\es$ be a conflict-driven \abbrses/.
Then, for $p\neq p'$ if $\Theta_\es(p)=\Theta_\es(p') \Rightarrow p \# p'
\;\&\; \neg(p \#_\mu p')$, where $\Theta_\es$ is the mapping from the
events of $\es$ to the events of $\Theta(\es)$ (definition \ref{def:trans-ses-pes-}).
\end{lemma}
\begin{proof}
Suppose $\Theta_\es(p)=\Theta_\es(p')=e$, then $p=\lceil e \rceil_v$, $p'=\lceil e \rceil_{v'}$. Let $v_0$ be the subset of $v$ s.t.\ $v_0\vdash_{min} e$ and similarly, let $v_1$ be the subset of $v'$ s.t.\ $v_1 \vdash_{min} e'$. By the stability axiom it follows that $v_0 \cup v_1 \notin Con$. It is then obvious that  $p \# p'$. Since $\es$ is conflict-driven, it follows that $\exists e_0 \in v_0, e_1 \in v_1.\; e_0 \# e_1$, and therefore, $\lceil e_0 \rceil_v \# \lceil e_1 \rceil_{v'}$  and since  $\lceil e_0 \rceil_v < \lceil e \rceil_v$ and $\lceil e_0 \rceil_{v'} < \lceil e \rceil_{v'}$ it follows that $\neg (p \#_\mu p')$.
\end{proof}

\begin{proof}[Proof of Theorem \ref{thm:bc-iso-bc-}]
Let $\vartheta = \tilde\theta_\es$ (definitions
\ref{def:trans-ses-pes-}, \ref{def:hatted}).
First consider configuration $v$ of  $\hat{\es}$ and $v'=\vartheta(v)$ of  $\mathcal{E}$. Since configurations of $\hat\es$ and $\es$ are isomorphic, it is clear that $\exists e\in E\setminus v.\; v \cup \{e\} \in \confs(\es)\Leftrightarrow \exists e'\in E\setminus v'.\; v' \cup \{e'\} \in \confs(\es)$. In other words, every initial event in future of $v$ has an associated initial event in future of $v'$ and vice versa.
Let $\hat\es^v_0$ and  $\es^{v'}_0$ represent the initial events of each structure, respectively. Then we show that $\vartheta$ yields a bijection between $\hat\es^v_0$ and  $\mathcal{E}^{v'}_0$.

Suppose $e,e'\in \hat\es^v_0$ and $\vartheta(e)=\vartheta(e')=e''$. By
lemma \ref{lem:morph-conf} it follows that $e \# e'$ and $\neg e
\#_\mu e'$, i.e.\;there is a conflict in their past. But this is a
contradiction as they are both initial events in future of a
configuration which is conflict-free. As shown above, every initial
event in future of $v'$ has an associated event in future of $v$,
therefore, $\vartheta$ (applied to initial events $\hat\es^v$) is onto
the initial events of $\es^{v'}$. It then follows that $\vartheta$ yields a bijection between the events of $\hat\es^v_0$ and $\mathcal{E}^{v'}_0$, making them isomorphic. 

Furthermore, as we are dealing with the initial events that can occur in future of $v'$ i.e.\ immediately after $v'$, the immediate conflict relation among the events of $\mathcal{E}^{v'}_0$ is resolved, in the sense that it does not depend on any configuration in $\es_0^{v'}$, and therefore, is obviously compatible with the immediate conflict relation of $\hat\es^v_0$. Thus, the branching cells of $\hat\es^v_0$ and  $\mathcal{E}^{v'}_0$ are isomorphic, which implies the branching cells of $\hat\es$ and those of $\es$ are isomorphic.  
\end{proof}

\end{document}